\documentclass[aps,superscriptaddress,prb,amsmath,amssymb,reprint,longbibliography]{revtex4-2}

\usepackage[USenglish]{babel}
\usepackage[T1]{fontenc}
\usepackage[utf8]{inputenc}
\usepackage{graphicx}
\usepackage{upgreek}
\usepackage[unicode]{hyperref}
\usepackage{multirow}
\usepackage{verbatim} 
\usepackage[usenames, dvipsnames]{color} 
\usepackage[normalem]{ulem} 
\usepackage{amsthm} 
\usepackage{verbatim} 
\usepackage[usenames, dvipsnames]{color} 
\usepackage{soul}
\usepackage{fixltx2e}
\usepackage{booktabs} 
\usepackage{siunitx}  
\usepackage{titlesec}
\usepackage{tikz}
\usepackage{float}
\usetikzlibrary{shapes.geometric}
\newtheorem{lemma}{Lemma}

\definecolor{pinocol}{rgb}{.1,0.8,0.1}
\definecolor{pinocolbg}{rgb}{.1,0.8,0.2}

\titlespacing*{\section}{0pt}{12pt}{6pt} 
\begin{document}


\setlength{\bibsep}{6pt} 
	
\title{33 Gbit/s source-device-independent quantum random number generator based on heterodyne detection with real-time FPGA-integrated extraction}

\date{\today}

\author{Marius Cizauskas}
\email{email: marius.cizauskas@tu-dortmund.de}
\affiliation{Experimentelle Physik 2, Technische Universit\"at Dortmund, 44221 Dortmund, Germany}
\affiliation{Department of Physics and Astronomy, University of Sheffield, Sheffield, UK}

\author{Hamid Tebyanian}
\affiliation{School of Physical and Chemical Sciences, Queen Mary University of London, London, UK}

\author{A. Mark Fox}
\affiliation{Department of Physics and Astronomy, University of Sheffield, Sheffield, UK}

\author{Manfred Bayer}
\affiliation{Experimentelle Physik 2, Technische Universit\"at Dortmund, 44221 Dortmund, Germany}

\author{Marc Assmann}
\affiliation{Experimentelle Physik 2, Technische Universit\"at Dortmund, 44221 Dortmund, Germany}

\author{Alex Greilich}
\email{email: alex.greilich@tu-dortmund.de}
\affiliation{Experimentelle Physik 2, Technische Universit\"at Dortmund, 44221 Dortmund, Germany}

\begin{abstract}
We present a high-speed continuous-variable quantum random number generator (QRNG) based on heterodyne detection of vacuum fluctuations. The scheme follows a source-device-independent (SDI) security model in which the entropy originates from quantum measurement uncertainty and no model of the source is required; security depends only on the trusted measurement device and the calibrated discretization, and thus remains valid even under adversarial state preparation. The optical field is split by a 90$^\circ$ optical hybrid and measured by two balanced photodiodes to obtain both quadratures of the vacuum state simultaneously. The analog outputs are digitized using a dual-channel 12-bit analog-to-digital converter operating at a sampling rate of 3.2 GS/s per channel, and processed in real time by an FPGA implementing Toeplitz hashing for randomness extraction. The quantum-to-classical noise ratio was verified through calibrated power spectral density measurements and cross-checked in the time domain, confirming vacuum-noise dominance within the 1.6 GHz detection bandwidth. After extraction, the system achieves a sustained generation rate of $R_{\rm net}= 33.92~\mathrm{Gbit/s}$ of uniformly distributed random bits, which pass all NIST and Dieharder statistical tests. The demonstrated platform provides a compact, FPGA-based realization of a practical heterodyne continuous-variable source-independent QRNG suitable for high-rate quantum communication and secure key distribution systems.
\end{abstract}

\maketitle

\section{Introduction}
Random numbers serve as fundamental building blocks across numerous fields, from Monte Carlo simulations in scientific research~\cite{Martin_2021} to quantum key distribution based cryptographic protocols~\cite{Alkhazragi_2023} of which the purpose is to provide unbreakable secure encrypted communication~\cite{Baar_2023}. In cryptographic applications, the quality and unpredictability of random numbers directly determine the security of encryption schemes, since they are used in one-time pad encryption, where the key length must match the message length and cannot be reused~\cite{Alkhazragi_2023}. While classical random number generators based on deterministic algorithms can satisfy many computational needs, they fundamentally cannot provide the true unpredictability required for security-critical applications, as their outputs are ultimately predictable given sufficient computational resources and knowledge of the algorithm's internal state~\cite{Garipcan_2021}.

Quantum mechanics offers a fundamental solution to this limitation through the intrinsic randomness present in quantum measurements. QRNGs exploit the fundamental uncertainty principle of quantum mechanics to produce truly unpredictable random numbers, providing a level of security that classical systems cannot achieve. The measurement of quantum observables yields outcomes that are fundamentally non-deterministic, even with complete knowledge of the quantum state prior to measurement. There are several established QRNG methods, such as measuring single photons going through a 50:50 beamsplitter~\cite{Jennewein_1999}, measurements of laser phase fluctuations~\cite{Yang_2016}, amplified spontaneous emission in fiber amplifiers~\cite{Williams_2010}, vacuum fluctuations~\cite{Bai_2021, Zheng_2018} and etc. QRNGs are further divided into two main categories. There are discrete variable (DV) QRNGs, where the random numbers are formed by quantum measurements that yield binary results, such as single-photon measurements, and there are continuous variable (CV) QRNGs, where continuous quantum mechanical properties are measured.

In recent times, QKD systems have been improving and are operating in the GHz range~\cite{Grnenfelder_2020, Takesue_2007}. With the increasing QKD operating frequency, higher bandwidth QRNGs are needed to match the data transmission rates. There are already implementations showing bandwidths of more than a hundred Gbits/s~\cite{Bruynsteen_2023, Ng_2023}. However, these systems are not fully secure. The sources or the measurement devices can be tampered, which reduces the quality of QRNG, allowing for potential attacks~\cite{Thewes_2019}. There are device-independent (DI) QRNGs, which assume that both source and measurement devices cannot be trusted and apply self-testing algorithms to ensure that the QRNG cannot be tampered with. However, due to how complicated DI schemes are, the achievable bandwidths range is only up to kbits/s~\cite{Liu_2019}. A good trade-off is a source device-independent (SDI) QRNG, where only the measurement device is assumed to be completely trusted. There are already implementations of SDI QRNGs with bandwidth rates in the 10-30 Gbits/s range~\cite{Avesani_2018, Shrivastava_2025, Guo_2025}. However, in those cases, the processing is performed offline, making the effective bandwidth significantly lower. Therefore, higher bandwidths can be achieved by performing the processing fully on FPGA.

In this paper, we present a SDI QRNG based on heterodyne measurements of vacuum fluctuations, with acquisition, extraction and data transfer done fully on the FPGA. By using an FPGA together with a dual-channel analog-to-digital converter (ADC12DJ5200RFEVM) operating at a sampling rate of $f_s = 3.2$ GHz per channel (up to 5.2 GHz maximum), both heterodyne detection channels are digitized in parallel to generate random numbers. A net extracted throughput of $R_{\rm net}= 33.92~\mathrm{Gbit/s}$ is achieved on this set up. In the subsequent sections we cover the experimental setup, including the FPGA implementation, the results of the QRNG and its statistical tests.
\begin{figure*}
\centering
\includegraphics[width=17.5cm]{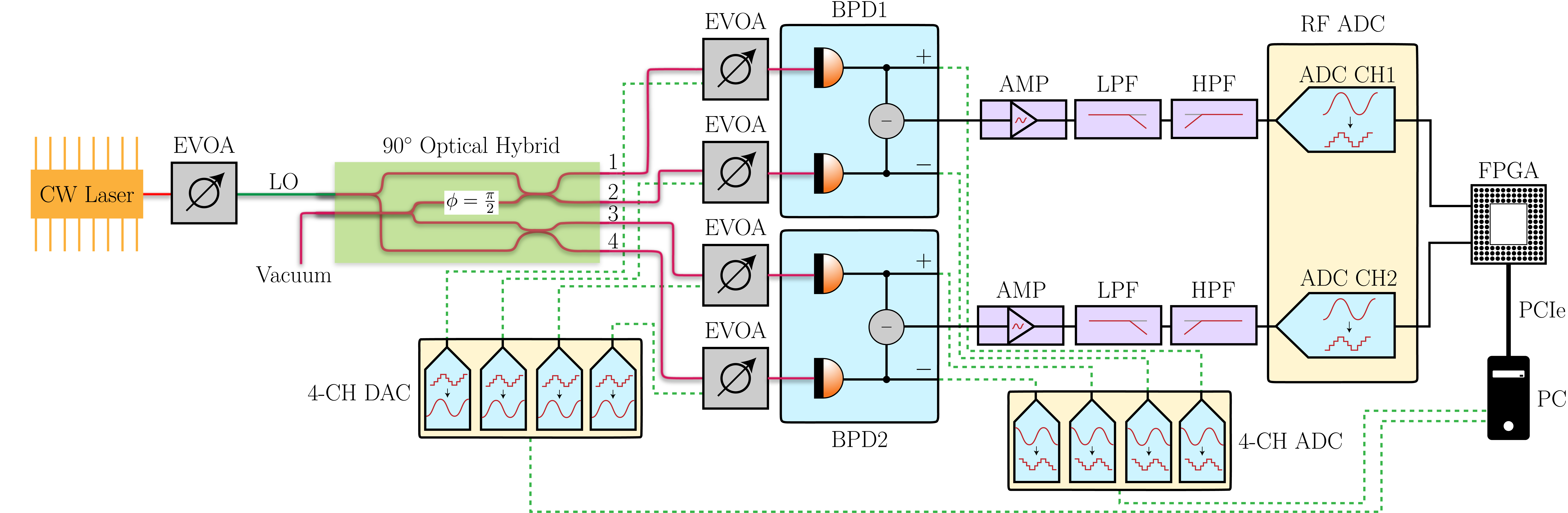}
\caption{Schematic diagram of the vacuum fluctuation-based quantum random number generator using heterodyne detection. The setup consists of a laser source followed by a variable optical attenuator (VOA) for power control and stabilization. The laser output is mixed with vacuum fluctuations in a 90$^\circ$ optical hybrid, which simultaneously measures both X and P quadratures of the electromagnetic field. Each of the four hybrid outputs are individually controlled by electronically variable optical attenuators (EVOAs) with electrical control inputs before being detected by balanced photodiode detectors (BPD1 and BPD2). The analog electrical signals from both detectors are first filtered by a low-pass and high-pass filters to filter out as much classical electrical noise as possible and then digitized by a 2-channel analog-to-digital converter (ADC) and processed in real-time by a field-programmable gate array (FPGA) connected via an FMC (FPGA Mezzanine Card) interface.}
\label{fig:qrng_setup}
\end{figure*}

\section{Experimental details}
\subsection{General Setup}
The experimental setup, shown in Fig.~\ref{fig:qrng_setup}, implements a heterodyne detection scheme for vacuum fluctuation-based quantum random number generation. The system is built on a continuous-wave laser source {LPSC-1550-FC with CLD1010LP laser diode driver from Thorlabs} operating at wavelength $\lambda=1550\,nm$ with output powers up to 30\,mW. The laser output serves as the local oscillator (LO). It passes through a variable optical attenuator (VOA) {EVOA1550A} to precisely control the LO power and to ensure its stability. The VOA is set to limit the detector-plane LO power to $\le 15$\,mW for the variance–slope calibration range.

The LO is input to a 90$^\circ$ optical hybrid {Optoplex model HB-C0AFAS066} operating in the telecom C-band. This device implements a 2×4 optical hybrid coupler that mixes the LO $E_{LO}$ with the reference signal (vacuum in this case) $E_{vac}$ from the environment to generate four quadratural states in the complex-field space.

Each of the four hybrid outputs is individually controlled by electronically variable optical attenuators (EVOAs) Thorlabs V1550A. These devices enable real-time gain balancing which allows for 90$^\circ$ optical hybrid outputs to be properly balanced. This ensures equal power distribution to maintain the quadrature measurement accuracy and compensate for any asymmetries in the optical path.

The four EVOA outputs are paired and fed into two Thorlabs PDB480C-AC balanced photodiode detectors (BPDs). Each BPD consists of a pair of matched photodiodes operated in differential mode. The 90$^\circ$ optical hybrid routes the LO and the signal (vacuum at the signal port) so that ports (1,2) and (3,4) form two balanced pairs with a relative phase shift of 180$^\circ$ within each pair and 90$^\circ$ between the two pairs. With the diodes in differential mode, these two balanced differences implement a simultaneous measurement of orthogonal quadratures (heterodyne). Denoting the differential voltages by $A$ and $B$ (for the two BPDs), we model
\begin{equation}
\begin{aligned}
A &\equiv V_X \;=\; \kappa_X \sqrt{P_{\rm LO}^{(\mathrm{det})}}\, X \;+\; n_{e,X},\\
B &\equiv V_P \;=\; \kappa_P \sqrt{P_{\rm LO}^{(\mathrm{det})}}\, P \;+\; n_{e,P},
\end{aligned}
\label{eq:hybrid}
\end{equation}
where $\kappa_{X/P}$ are responsivities [V$/\sqrt{\rm W}$], $X,P$ are heterodyne quadrature outcomes with $\langle X\rangle=\langle P\rangle=0$ and ${\rm Var}(X)={\rm Var}(P)=\tfrac12$, and $n_{e,X/P}$ are electronics noises (V). 

Hence the shot-noise slopes satisfy
\begin{equation}
m_X = \frac{\kappa_X^{\,2}}{2},\qquad m_P = \frac{\kappa_P^{\,2}}{2}.
\end{equation}
Consequently, the variances scale linearly with LO power,
\begin{equation}
\begin{aligned}
\operatorname{Var}(A)= m_X P_{\rm LO}^{(\mathrm{det})} + \sigma_{e,X}^2,\\
\operatorname{Var}(B)= m_P P_{\rm LO}^{(\mathrm{det})} + \sigma_{e,P}^2,
\end{aligned}
\end{equation}

where $\sigma_{e,X}^2$, $\sigma_{e,P}^2$ are the measured voltage variances from the photodiodes. These quantities underpin the slope–offset calibration used in Sec.~\ref{sec:calibdelta}.

The analog electrical signals from both BPDs, representing the $X$ and $P$ quadratures, first go through a ZX60-P105LN+ 15 dB amplifier to ensure that most of the ADC input range is covered, then they go through a series connected low-pass filter (LPF) and high-pass filter (HPF). The LPF is a SLP-1650+ with a bandwidth of DC-1650 MHz and the HPF is a SHP-200+ with a bandwidth of 185-3000 MHz, yielding an effective bandpass filter with a bandwidth of 185-1650\,MHz. The upper limit was chosen at 1650\,MHz since the diodes have a bandwidth of up to 1600 MHz and the sampling rate is 3.2\,GHz, which puts the bandwidth slightly outside the Nyquist frequency; the \emph{LPF} nominal passband is DC-1400\,MHz with a -3\,dB point at 1650\,MHz. The lower bound was chosen at 185 MHz to reduce the classical lower frequency noise originating in the diodes. Lastly, the diode outputs are digitized by a 2-channel ADC ADC12DJ5200RFEVM with a sampling rate of $f_s=3.2$\,GHz set for both channels and a resolution of 12 bits. The digitized quadrature data streams are transmitted to a field-programmable gate array (FPGA) Virtex Ultrascale+ VCU118 development board via a high-speed FMC (FPGA Mezzanine Card) interface. The FPGA implements a real-time randomness extraction algorithm via Toeplitz hashing to generate provably random bit sequences from the quantum vacuum fluctuations by flattening the Gaussian distribution of the data.

\subsection{Randomness Generation and Extraction}
Heterodyne measurement forms the foundation of SDI quantum random number generation by providing cryptographic security without requiring assumptions about the quantum source. Unlike homodyne detection that measures a single quadrature, heterodyne detection simultaneously measures both X and P quadratures of the electromagnetic field, enabling higher generation rates and enhanced security. The source device independence property emerges from the fundamental structure of the heterodyne measurement itself. In an SDI framework, an adversary may have complete control over the quantum source, manipulating it to maximize their ability to predict measurement outcomes. However, the security is guaranteed by exploiting the properties of the Positive Operator Valued Measurement (POVM), which is inherent to heterodyne measurements~\cite{Baragiola_2017}.

\begin{figure*}[]
\centering
\includegraphics[width=18cm]{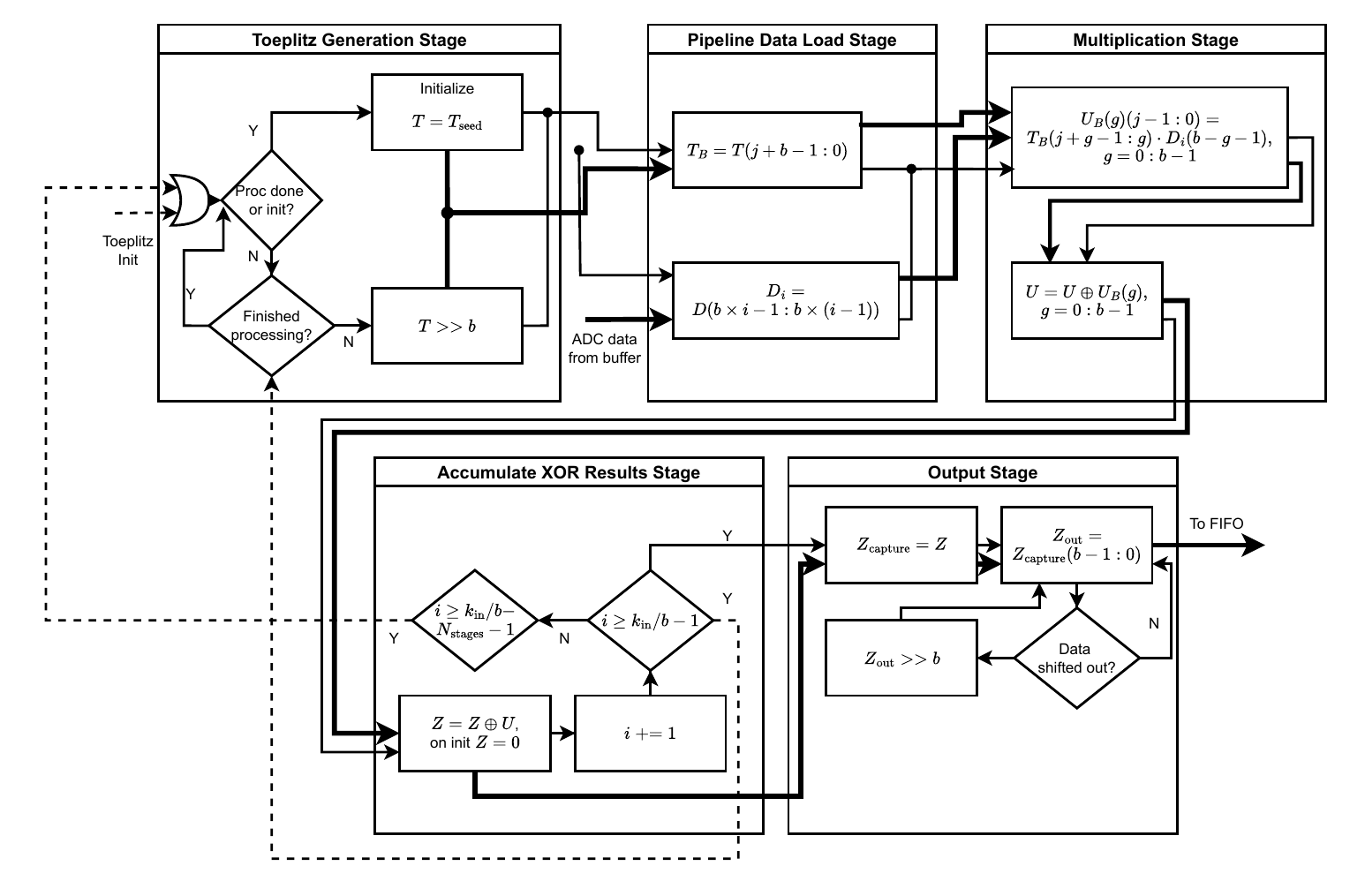}
\caption{The pipelined Toeplitz extraction scheme. The scheme consists of separate pipeline stages. The dashed arrows indicate binary signals, the normal lines indicate the flow of the program and the bold, larger lines indicate flow of data. The square blocks represent operations and the diamond blocks represent conditionals.}
\label{fig:extraction}
\end{figure*}

Minimum entropy quantifies the worst-case unpredictability of a random variable in the presence of an adversary.
For a random variable $X$ with probability distribution $p(x)$,
\begin{equation}
    H_{min}(X) = -\log_2\big(\max_{x}p(x)\big).
    \label{eq:hmin}
\end{equation}
In QRNG security we use the smooth conditional min-entropy $H_{\min}^{\varepsilon_s}(\cdot|E)$ where $E$ denotes adversary's (quantum) side information; we write $H$ for classical protocol history when needed. Let $Z$ denote the discretised heterodyne outcome (ADC bin index). Smooth conditional min-entropy evaluates non-smooth entropies on the optimal state within an $\epsilon$-neighbourhood of the measured state, where $\epsilon$ is the security parameter, while the per-round heterodyne discretization bound is state-independent and holds \emph{unsmoothed}~\cite{Avesani_2018}:
\begin{equation}
   H_{min}(X|E) \ge
    -\log_2\Bigl(\min\bigl\{\tfrac{\delta'_X\delta'_P}{2\pi},\,1\bigr\}\Bigr),
    \label{eq:hmim_side}
\end{equation}

The heterodyne POVM is $\Pi(\alpha)=|\alpha\rangle\langle\alpha|/\pi$, so $Q_\rho(\alpha)=\langle\alpha|\rho|\alpha\rangle/\pi$ and $\sup_\rho Q_\rho=1/\pi$.
With ${\rm Var}(X)={\rm Var}(P)=\tfrac12$ we have $\alpha=(X+iP)/\sqrt{2}$ and $d^2\alpha=dX\,dP/2$.
If the ADC defines rectangles $\{R_{mn}\}$ of area $\delta'_X\delta'_P$ in $(X,P)$, then $\mathrm{Area}_\alpha(R_{mn})=\delta'_X\delta'_P/2$.
Hence (using $d^2\alpha=dX\,dP/2$ for vacuum units ${\rm Var}(X)={\rm Var}(P)=1/2$)
\begin{align}
p_{\max} := \max_{m,n}\int_{R_{mn}} Q_\rho(\alpha)\,d^2\alpha
      \;\le\; \min\Bigl\{\tfrac{\delta'_X\delta'_P}{2\pi},\,1\bigr\},
\label{eq:pmax}
\end{align}
and
\begin{align}
h_{\min}^{(1)} \;\ge\; -\log_2\Bigl(\min\bigl\{\tfrac{\delta'_X\delta'_P}{2\pi},\,1\bigr\}\Bigr).
\label{eq:h1}
\end{align}
For any (possibly adversarial) preparation and any past transcript,
\begin{equation}
\max_{m,n}\Pr\big[(X,P)\in R_{mn}\,|\,\text{history},\mathcal{E}\big]\le       \min\Bigl\{\tfrac{\delta'_X\delta'_P}{2\pi},\,1\Bigr\},
\end{equation}
and thus $h_{\min}^{(1)}\ge -\log_2\Bigl(\min\{\tfrac{\delta'_X\delta'_P}{2\pi},\,1\}\Bigr)$ holds per round without IID assumptions and is adaptivity-robust.

\begin{lemma}
\label{lem:opnorm}
Let $M_R=\int_R |\alpha\rangle\langle\alpha|\,\frac{d^2\alpha}{\pi}$ be the heterodyne POVM element for a measurable bin $R\subset\mathbb{C}$. Then
\begin{equation}
    \lVert M_R\rVert_\infty \;\le\; \min\Bigl\{\tfrac{\mathrm{Area}(R)}{\pi},\,1\Bigr\}.
\end{equation}
\end{lemma}
\begin{proof}
For any unit vector $|\psi\rangle$, $\langle\psi|M_R|\psi\rangle=\int_R Q_\psi(\alpha)\,d^2\alpha$ with
$Q_\psi(\alpha)=|\langle\psi|\alpha\rangle|^2/\pi\le 1/\pi$; hence
$\langle\psi|M_R|\psi\rangle\le\mathrm{Area}(R)/\pi$.
Completeness of the POVM gives $\int_{\mathbb{C}}|\alpha\rangle\langle\alpha|\,\frac{d^2\alpha}{\pi}=\mathbb{I}$,
so $0\le M_R\le\mathbb{I}$ and $\lVert M_R\rVert_\infty\le 1$.
Taking the tighter of the two bounds yields the stated minimum.
\end{proof}

Boundary (rail-adjacent but non-clipped) bins are treated as truncated rectangles with their actual area used in \eqref{eq:pmax}. ADC \emph{clipping} codes result from voltages outside the input range $[-V_{\rm pp}/2,\,V_{\rm pp}/2]$ of the ADC. They correspond to semi-infinite regions and are excluded. Dropping clips is a post-selection on the event $\{\text{no clip}\}$ of probability $1-f_{\rm clip}$, which renormalises the per-round distribution. This conditioning is secure because clipping is determined solely by the trusted ADC's range. Consequently, the single-round bound for the post-selected data becomes
\begin{equation}
h_{\min,\mathrm{no\;clip}}^{(1)} \;\ge\;
-\log_2\Bigl(\min\bigl\{\tfrac{\delta'_X\delta'_P}{2\pi(1-f_{\rm clip})},\,1\bigr\}\Bigr).
\label{eq:h1_noclip}
\end{equation}
and the extraction length obeys
\begin{equation}
\begin{aligned}
\ell\;\le\;n_{\rm valid}\,h_{\min,\;\mathrm{no\;clip}}^{(1)}\;-\;2\log_2\frac{1}{\varepsilon_{\rm PA}},\\
n_{\rm valid}=n(1-f_{\rm clip}),
\end{aligned}
\end{equation}
where $n$ is the total number of rounds entering extraction.

\subsection*{Calibration and vacuum-unit resolutions}
\label{sec:calibdelta}
 The bound depends only on the discretization induced by the apparatus. We determine the raw resolutions $\delta_{X/P}$ and the conservative resolutions $\delta'_{X/P}$ from a slope–offset calibration and the ADC’s effective code width as follows.
\begin{equation}
\begin{aligned}
\sigma^2_{V,X}(P_{\rm LO}) = m_X P_{\rm LO} + c_X,\\
\sigma^2_{V,P}(P_{\rm LO}) = m_P P_{\rm LO} + c_P,
\label{eq:varfit}
\end{aligned}
\end{equation}
with $m_{X/P}$ in $\mathrm{V}^2/\mathrm{W}$ and $c_{X/P}\ge0$ in $\mathrm{V}^2$.
To exclude LO relative-intensity-noise leakage due to residual imbalance, we fit
\begin{equation}
\begin{aligned}
\sigma^2_{V,X}(P_{\rm LO}) = m_X P_{\rm LO} + c_X + \eta_X P_{\rm LO}^2,\\
\sigma^2_{V,P}(P_{\rm LO}) = m_P P_{\rm LO} + c_P + \eta_P P_{\rm LO}^2,
\end{aligned}
\end{equation}
and test $\eta_{X/P}=0$ via lack-of-fit. We found $|\eta_{X/P}|P_{\rm LO,\max}\ll m_{X/P}$ within the calibration range; thus linear shot-noise slopes are valid. If $\eta\neq 0$, its effect is absorbed by increasing $\gamma_{X/P}$ in \eqref{eq:deltas-infl}.

Here $c_X=\sigma_{e,X}^2$ and $c_P=\sigma_{e,P}^2$, consistent with the detector model in \eqref{eq:hybrid}. Since ${\rm Var}(X)={\rm Var}(P)=\tfrac12$ for vacuum, the conversion factors to vacuum units are:
\begin{equation}
\alpha_X=\sqrt{2 m_X P_{\rm LO}^{(\mathrm{det})}},\quad \alpha_P=\sqrt{2 m_P P_{\rm LO}^{(\mathrm{det})}},
\label{eq:alphas}
\end{equation}
here $P_{\rm LO}^{(\mathrm{det})}$ denotes the \emph{per-BPD} LO power at the photodiodes \emph{after} the 90$^\circ$ hybrid and EVOA balancing (i.e., the sum incident on the two diodes of a given BPD). All power readings used in \eqref{eq:alphas} are taken at this plane. We define per-sample vacuum-unit coordinates by
\begin{equation}
\hat X := \frac{A}{\alpha_X},\qquad \hat P := \frac{B}{\alpha_P},
\end{equation}
so that for vacuum $\operatorname{Var}(\hat X)=\operatorname{Var}(\hat P)=\tfrac12$. With ADC peak-to-peak range $V_{\rm pp}=0.5$\,V after the analog front end, define the effective code width
\begin{equation}
\Delta V_{\rm eff}=\frac{V_{\rm pp}}{2^{\,n_{\rm ENOB}}},
\label{eq:deltaVeff}
\end{equation}
where $n_{\rm ENOB}=10.2$ (X channel) and $10.3$ (P channel) is the measured effective number of bits of the ADC (sine-wave histogram method, averaged over 185–1600\,MHz). Datasheet ENOB variation is absorbed into $\beta$. The raw resolutions are
\begin{equation}
\delta_X=\frac{\Delta V_{\rm eff}}{\alpha_X},\qquad \delta_P=\frac{\Delta V_{\rm eff}}{\alpha_P}.
\label{eq:deltas-raw}
\end{equation}
Let $\sigma_{X}^2$ and $\sigma_{P}^2$ denote the measured quadrature variances expressed in vacuum units after applying the gains $\alpha_{X/P}$ from \eqref{eq:alphas}.
To conservatively include excess noise or calibration uncertainty, define
\begin{equation}
\begin{aligned}
\gamma_X:=\max\Big\{1,\sqrt{\sigma_X^2/(1/2)}\Big\},\\
\gamma_P:=\max\Big\{1,\sqrt{\sigma_P^2/(1/2)}\Big\},
\end{aligned}
\label{eq:inflate-factors}
\end{equation}
\begin{equation}
\delta'_X=\gamma_X\,\delta_X,\qquad \delta'_P=\gamma_P\,\delta_P.
\label{eq:deltas-infl}
\end{equation}
More generally, the single-shot bound reduction is $\log_2(\gamma_X\gamma_P)$. With the measured $(\sigma_X^2,\sigma_P^2)$, the single-shot penalty is
\begin{align}
\Delta h_{\min}^{(1)}&=\log_2(\gamma_X\gamma_P) \nonumber\\
&=\tfrac12\log_2\!\bigl(\max\{1,2\sigma_X^2\}\bigr)
  +\tfrac12\log_2\!\bigl(\max\{1,2\sigma_P^2\}\bigr).
\end{align}
For $(\sigma_X^2,\sigma_P^2)=(0.616,0.647)$ this gives $\Delta h_{\min}^{(1)}\approx 0.336~\text{bits/round}$.
In the symmetric case $\sigma_X^2=\sigma_P^2=\tfrac12(1+\xi)$, where $\xi$ is the relative excess noise above the vacuum variance, one has $\gamma_X=\gamma_P=\sqrt{1+\xi}$ and the reduction equals $\log_2(1+\xi)$. The guard $\gamma\ge1$ prevents artificial entropy gain if $\sigma^2<1/2$ due to calibration uncertainty.
The empirical cross-quadrature correlation was $|\rho_{XP}|={0.012}\pm{0.004}$. We apply a fixed orthonormal (unit-Jacobian) rotation $R$ (eigenbasis of the covariance matrix) to decorrelate, yielding $|\rho_{XP}^{\rm (out)}|<10^{-4}$ and no Jacobian penalty enters the bound. The resulting correlation penalty would be $\Delta h_{\rm corr}=-\frac12\log_2(1-(\rho_{XP}^{\rm (out)})^2)<8\times10^{-5}$ bits, confirming it is negligible.

We propagate uncertainties in $(m_{X/P},c_{X/P},V_{\rm pp},n_{\rm ENOB})$ to $(\delta'_X,\delta'_P)$ and to $h_{\min}^{(1)}$ by first-order error propagation; error bars for $\delta'_{X/P}$ and the net rate include these calibration contributions.

\begin{figure*}[]
\centering
\includegraphics[width=18cm]{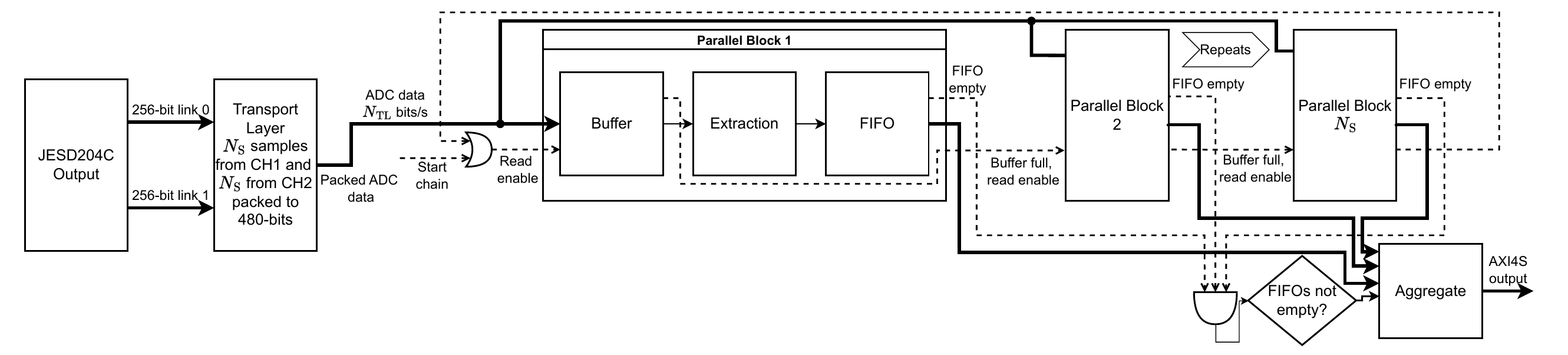}
\caption{The overall schematic of the FPGA design. All parallel blocks contain the same functions and I/O, only the first block is fully drawn to make the diagram more compact. The dashed arrows indicate binary signals, the normal lines indicate the flow of the program and the bold, larger lines indicate flow of data.}
\label{fig:parallel}
\end{figure*}

The security proof relies on the fact that for any quantum state, the Husimi function $Q_{\rho}(q + ip)$ is upper bounded by $1/\pi$. This fundamental quantum mechanical constraint ensures that even if an adversary prepares an optimal quantum state to maximize their guessing probability, the conditional min-entropy remains bounded by the measurement resolution alone~\cite{Avesani_2018}. This bound is independent of the actual quantum state prepared by the potentially malicious source, providing information-theoretic security guarantees. Practical implementations must account for device imperfections that can compromise this security. Imbalanced heterodyne detection, arising from non-ideal beam splitter ratios or photodiode efficiency mismatches, introduces local oscillator fluctuation that contributes excess noise. Under imbalanced conditions, LO fluctuation cannot be eliminated and affects the calibration of vacuum fluctuations~\cite{Li_23}. Without proper consideration of these fluctuations, the extractable randomness can be overestimated, creating security vulnerabilities. For this reason, we use additional EVOAs at the 90$^\circ$ optical hybrid, since it does not show perfect 50:50 splitting ratio.

As mentioned before, upon initial measurements of this system, the resulting distribution is Gaussian, which makes it easier to predict the generated random numbers. For this reason, a randomness extractor has to be used to flatten the distribution and extract the QRNG which is mixed with classical noise sources. One of the more commonly used extractors is Toeplitz hashing~\cite{Bai_2021, Zheng_2018, Nie_2015,Bruynsteen_2023,Avesani_2018, Guo_2025}, belonging to universal$_2$ family and hence a \emph{strong} extractor~\cite{Tomamichel_2011, Krawczyk_1994}: for a public, uniform, independent seed $S$, $(S,\mathrm{Ext}(Z^n,S)) \approx (S,U_j)$.

We reuse the same public seed $S$ across all parallel output blocks within one extractor invocation; Toeplitz hashing is a strong (universal$_2$) extractor, and our single-round heterodyne min-entropy bound is independent of $S$ and the classical transcript. Hence, by a hybrid/triangle-inequality argument applied block by block, the joint trace distance satisfies
\begin{equation}
\begin{aligned}
    &\big\|(S,\mathrm{Ext}(Z^{n}_{(1)},S),\ldots,\mathrm{Ext}(Z^{n}_{(B)},S),E)- \\
    &(S,U_{\ell_1},\ldots,U_{\ell_B},E)\big\|_{\mathrm{tr}}
\;\le\;\sum_{i=1}^{B}\big(\varepsilon_s^{(i)}+\varepsilon_{\rm PA}^{(i)}\big)
\end{aligned}
\end{equation}
where $\varepsilon_s^{(i)}$ is the smoothing error of a block and $\varepsilon_{\rm PA}^{(i)}$ is the privacy amplification error of a block. In this way we allocate $(\varepsilon_s^{(i)},\varepsilon_{\rm PA}^{(i)})$ so that the sum over blocks equals the advertised $\varepsilon_{\rm sec}$.

Toeplitz extraction multiplies a $k_{\rm in}$-bit input by a $j\times k_{\rm in}$ Toeplitz matrix specified by a $(j+k_{\rm in})$-bit seed. Here $k_{\rm in}=n_{\rm eff}\,b$ with $b=24$ bits/round (two 12-bit ADCs), and $j=\ell$ is the extracted length from the min-entropy bound, all arithmetic is over $\mathbb{F}_2$.

The heterodyne bound is state- and history-independent: for every retained round $t\le n_{\rm eff}$ and any classical transcript $H$,
\begin{equation}
    p_{\max}^{(t)}\;\le\;\min\Bigl\{\tfrac{\delta'_X\delta'_P}{2\pi},\,1\Bigr\}.
\end{equation}

Because the trusted device applies the same per-round POVM $\{M_z\}$ to each temporal mode and does not adapt to $H$, one has for any cq state with arbitrary inter-round correlations
\begin{equation}
\begin{aligned}
&H_{\min}^{\varepsilon_s}(Z^{n_{\rm eff}}|E,H)\;\ge\; n_{\rm eff}\,h_{\min}^{(1)},\\
&h_{\min}^{(1)}=-\log_2\Bigl(\min\bigl\{\tfrac{\delta'_X\delta'_P}{2\pi},1\bigr\}\Bigr).
\end{aligned}
\end{equation}
where $n_{\rm eff}$ is the number of temporal modes on which the trusted POVM $\{M_z\}$ acts identically.

We report integrated autocorrelation time $\tau_{\mathrm{int}}$ as a diagnostic of classical correlations in the sampled voltage stream. A conservative way to enforce approximate factorization is to define rounds on a decimated stream with spacing $D$ samples (e.g.\ $D\ge\lceil\tau_{\rm int}\rceil$), leading to $n_{\rm eff}=n/D$ and an effective sampling rate $f_s^{(\mathrm{eff})}=f_s/D$. In our reported security analysis and rates, we instead apply a fixed unit-Jacobian orthonormal transform on blocks before discretization, so that the $n$ samples are mapped to $n$ orthonormal temporal modes measured by the same POVM $\{M_z\}$; in this case we take $n_{\rm eff}=n$ and $f_s^{(\mathrm{eff})}=f_s$, and $\tau_{\rm int}$ is used purely as a diagnostic check that residual correlations after extraction are within statistical error.
The analog front-end and sampling chain are trusted and absorbed into the calibration that defines $\delta'_{X/P}$. After enforcing factorization as above, the product-POVM structure holds on the retained temporal modes and we use $H_{\min}(Z^{n_{\rm eff}}|E)\ge n_{\rm eff}\,h_{\min}^{(1)}$.

We separate the secrecy parameter from statistical-estimation error, thus we define
\begin{equation}
\varepsilon_{\rm sec}\;=\;\varepsilon_s+\varepsilon_{\rm PA}+\beta,
\end{equation}
where $\beta$ upper-bounds the probability that calibration parameters deviate outside their quoted confidence intervals.
Privacy amplification obeys
\begin{equation}
\ell \le H_{\min}^{\varepsilon_s}(Z^{n_{\rm eff}}|E) - 2\log_2\frac{1}{\varepsilon_{\rm PA}}.
\label{eq:l_hmin}
\end{equation}
Here $E$ denotes adversarial (quantum) side information. Since $H_{\min}^{\varepsilon_s}(\cdot|E)\ge H_{\min}(\cdot|E)$ for any $\varepsilon_s\in[0,1)$, our unsmoothed lower bound $H_{\min}(Z^{n_{\rm eff}}|E)\ge n_{\rm eff}\,h_{\min}^{(1)}$ also lower-bounds the smoothed quantity used in privacy amplification.

Here $\beta$ upper-bounds the probability that the calibration parameters $(m_{X/P},c_{X/P},\Delta V_{\rm eff},n_{\rm ENOB})$ deviate outside their quoted confidence intervals; it is set via conservative concentration bounds for the variance fits and datasheet limits. We set $\varepsilon_s=\varepsilon_{\rm PA}=\beta=2^{-64}$, yielding $\varepsilon_{\rm sec}=3\cdot 2^{-64}$.
If the data are output in $N_{\rm blk}$ extractor invocations (fresh or reused public seed), we allocate per-block parameters $(\varepsilon_s/N_{\rm blk},\,\varepsilon_{\rm PA}/N_{\rm blk},\,\beta/N_{\rm blk})$
so that a union bound guarantees the advertised global bound on $\varepsilon_{\rm sec}$.

\subsection{FPGA Implementation}
The FPGA receives 20 samples for each channel per 160 MHz clock cycle, or 480-bits of data, over 16 JESD204C data lanes.  This results in a total data rate of 76.8 Gbps at 3.2 GHz sample rate and 2-channel acquisition. Therefore, both pipelining and parallelization were applied to ensure that the data is processed in real-time. Fig.~\ref{fig:extraction} shows how the Toeplitz extraction scheme was implemented. In the figure shown, there are 5 pipelines stages, however, there is an additional initial stage for buffering ADC input, which is covered in further text.

Once there is data available in the input buffer, the pipeline starts with the Toeplitz generation stage. It can be seen that there is an OR gate for two signals, one which is triggered when Toeplitz extraction is done and another called Toeplitz init, which is only set when the FPGA is reset, so as to ensure proper Toeplitz matrix initialization on FPGA power on. If one of these signals is triggered, the initial Toeplitz vector signal $T$ is set from the independently generated seed. When neither of the two aforementioned signals is triggered, it indicates that Toeplitz extraction is running, therefore $T$ is shifted to the right by the combined bit-depth $b$. The right shift is needed because a block of the Toeplitz matrix $T_B$ is taken in the next pipeline data loading stage. A total of $j+b$ bits are taken for a block, which effectively corresponds to $b$ columns of the Toeplitz matrix, since, if the column size is $j$, the next column can be formed by $T_B(j+1:1)$ ($T_B$ here denotes a block of the Toeplitz matrix). When $T$ is shifted to the right by $b$, it results in taking further $b$ columns of the Toeplitz matrix. In the same block, the $b$-bit ADC word $D_i$ is taken every cycle from the buffer $D$.

Following is the multiplication stage, where first, the multiplication is done for the block of Toeplitz matrix, which is then stored in an intermediate two-dimensional vector $U_B$. As mentioned before, the numbers in the matrices are in binary, therefore mod 2 arithmetic applies. Therefore, the multiplication is calculated by applying AND bitwise operation on a single column of the Toeplitz matrix by a single bit of the ADC data $D_i$. This operation is done over a range of $0$ to $b-1$ so as to multiply all $b$ columns in the same cycle. Afterwards, a XOR reduction is done on $U_B$, which corresponds to the $b$-column multiplication sum, and is stored in $U$. After this, $U$ goes to the XOR accumulation stage, where each clock cycle $U$ is accumulated in the final result vector $Z$. This is done $k_{\rm in}/b$ times, until the ADC input buffer is fully processed. In this stage there are also two conditions: $i \ge k_{\rm in}/b - 1$ is used to signal when $Z$ is fully accumulated and it can be output, and $i \ge k_{\rm in}/b - N_{\mathrm{stages}}-1$, where $N_{\mathrm{stages}}$ is the number of stages between the Toeplitz generation and accumulation stages (in this case 3), which is necessary to ensure timely initialization of $T$, since the Toeplitz generation stage operates ahead of the accumulation stage due to pipelining created latency.

When the $i \ge k_{\rm in}/b - 1$ condition is satisfied, the output stage is started, where firstly $Z$ is captured to $Z_{\mathrm{capture}}$ to ensure that the extracted value is not overridden by further calculations in the pipeline. The data is output in pieces of $b$ bits for $j/b$ cycles by shifting $Z_{\mathrm{out}}$ to the right by $b$. Lastly, it should be mentioned that between all these stages there's a single clock cycle latency and that initially all stages are inactive until the first Toeplitz stage becomes active.

Fig.~\ref{fig:extraction} shows only the extraction diagram of a single block. If $f_{\mathrm{ADC}}$ is the ADC sampling frequency (per channel), the total input rate is $b\,f_{\mathrm{ADC}}$ bits/s, where $b=2\,n_{\rm ADC}=24$ bits/round. The FPGA receives $N_{\mathrm{S}}$ samples per cycle per channel, so the reference clock is $f_{\mathrm{ref}}=f_{\mathrm{ADC}}/N_{\mathrm{S}}$. A single block processes $b$ bits per clock cycle, i.e., $b\,f_{\mathrm{ref}}$ bits/s. Thus the number of parallel blocks required for full throughput is $N_{\mathrm{B}}=\frac{b f_{\mathrm{ADC}}}{b f_{\mathrm{ref}}}=N_{\mathrm{S}}$. In our setup $N_{\mathrm{S}}=20$, so 20 blocks are used to extract randomness in real time.

Fig.~\ref{fig:parallel} shows the full implementation of the FPGA design. Starting on the left, there is the JESD204C protocol output of the ADC, which goes into the transport layer, where first JESD204C lanes are appropriately mapped, then samples are extracted from the frames and finally the values are packed into 480-bit AXI4 Stream buses, where both CH1 and CH2 samples are interleaved in 24-bit words. This data is then sent to all the extraction blocks. All the blocks are interconnected by the buffering logic. The buffer is an additional pipeline stage part of the extraction schematic shown in Fig.~\ref{fig:extraction}. It consists of additional logic that ensures that the right amount of data is being read and to allow for proper operation with other parallel extraction blocks. More specifically, the buffer block additionally has a read enable input and a buffer full output. The buffer does not take any data until read enable is asserted and the buffer full output is asserted when the buffer is filled.

Such a design allows for daisy chaining the blocks and allow them to work in parallel. As seen in Fig.~\ref{fig:parallel} the read enables and the buffer full signals are chained together, with the exception of the first block, where there is an OR gate to the input of read enable. Initially, all of the blocks are inactive. Once ADC data becomes valid, the start chain signal connected to the OR gate of the first block read enable is asserted, which starts the ADC sample collection in the first block's buffer. Once the buffer is full, the buffer full signal is asserted and the collection starts in the second block and continues until the last block, where the buffer full signal is fed back to the first one, repeating the chain again. Finally, after extraction, the random bits are sent into a FIFO. The FIFO empty signals from each block are checked and when the condition of all FIFOs not being empty is satisfied, the FIFO outputs are read and are collated into a single vector signal.

To realize the full potential of the high bandwidth QRNG, the extracted data is sent over x16 PCIe to a computer. The PCIe interface on the FPGA runs with a clock speed of 250 MHz and 512-bit depth. To ensure proper clock domain crossing (CDC), the QRNG data first goes into an asynchronous FIFO, from which the PCIe interface reads. PCIe on the FPGA is handled by the XDMA IP provided by AMD. The computer that the FPGA is connected to runs Linux and uses the AMD provided XDMA driver to stream the data from the FPGA to the computer over an AXI4 Stream interface. There is an additional, not shown channel of the PCIe used to also transfer the raw, unextracted data to the computer. It is used both for security analysis and for monitoring the system during operation. The full implementation was successfully placed and routed without any timing issues. The used resources for a Toeplitz extraction for {$j = 1272$, $k = 2880$} and 20 parallel blocks can be seen in Table~\ref{tab:resource_usage}. The most significant usage is seen in the LUTs, which is expected, since even though the operations are not applied on a full matrix, the vector widths are still relatively large and therefore a lot of LUTs are needed to perform the bitwise operations over 20 independent extraction blocks. The matrix size could be slightly increased, however, at higher resource usage values the design has higher likelihood to fail implementation, since when usage is over 75 \%, congestion becomes an issue, making the placement and routing of logic challenging~\cite{amd2025_ug949}.

\begin{table}[h]
\centering
\begin{tabular}{|l|r|r|r|}
\hline
Resource & Used & Available & Utilization (\%) \\
\hline
LUT     & 596422 & 1182240 & 50.4 \\
LUTRAM  & 12697  & 591840  & 2.1 \\
FF      & 329457 & 2364480 & 13.9 \\
BRAM    & 197    & 2160    & 9.1 \\
DSP     & 20     & 6840    & 0.3 \\
IO      & 24     & 832     & 2.9 \\
GT      & 32     & 52      & 61.5 \\
BUFG    & 25     & 1800    & 1.4 \\
MMCM    & 1      & 30      & 3.3 \\
PCIe    & 1      & 6       & 16.7 \\
\hline
\end{tabular}
\caption{FPGA resource utilization for {$j = 1272$ and $k_{\rm in} = 2880$} Toeplitz matrix extraction algorithm with 20 parallel blocks.}
\label{tab:resource_usage}
\end{table}

\section{Results and Discussion}
In Fig.~\ref{fig:fft} we show the power spectral density (PSD) of both ADC channels when measuring with LO and without LO. The results are similar to other heterodyne and homodyne set-ups, where we see an LO-induced noise-power increase of $\approx 7.5~\mathrm{dB}$ (X) and $\approx 6.5~\mathrm{dB}$ (P)~\cite{Zheng_2018, Shrivastava_2025, Shen_2010}, indicating that quantum noise is present, not only the classical noise. In particular,
\begin{equation}
\begin{aligned}
\Delta_{\rm dB}^{(X)}=10\log_{10}\left(\frac{m_X P_{\rm LO}^{(\mathrm{det})}+c_X}{c_X}\right),\\
\Delta_{\rm dB}^{(P)}=10\log_{10}\left(\frac{m_P P_{\rm LO}^{(\mathrm{det})}+c_P}{c_P}\right),
\end{aligned}
\end{equation}
which follow from \eqref{eq:varfit}. With
$m_X=2.41{\times}10^{-1}\,\mathrm{V}^2/\mathrm{W}$, $c_X=4.35{\times}10^{-4}\,\mathrm{V}^2$,
$m_P=2.33{\times}10^{-1}\,\mathrm{V}^2/\mathrm{W}$, $c_P=5.60{\times}10^{-4}\,\mathrm{V}^2$,
and $P_{\rm LO}^{(\mathrm{det})}=8.4\,\mathrm{mW}$,
\begin{equation}
\Delta_{\rm dB}^{(X)} \approx 7.52~\mathrm{dB},\qquad
\Delta_{\rm dB}^{(P)} \approx 6.52~\mathrm{dB}.
\end{equation}
One difference compared to the other references is the dip in intensity we see below 0.2 GHz, which is caused by the HPF we have added which has a bandwidth of 185 to 3000 MHz. The resulting filter band (185–1650\,MHz) lies only slightly outside of Nyquist frequency for $f_s=3.2$\,GHz ($f_s/2=1.6$\,GHz), so residual aliasing into the analysis region is negligible. Another inconsistency is the observation of spikes in the spectrum. The spikes at the highest frequency are related to spurious codes of the ADC at $f_s/2$. The spikes at lower frequencies are related to electromagnetic interference resulting from surrounding electronics because they can be seen with LO off or on, indicating that their source is not related to the optics. These spikes are narrow in bandwidth and do not impact the security after extraction.

As a consistency check, we integrate the PSD difference (LO on minus LO off) over the analysis band (185–1600\,MHz) to obtain the quantum-to-classical noise ratio (QCNR), reported with uncertainty. This cross-check is consistent with the slope-based calibration used for $\delta'_{X/P}$ and is not used directly in the entropy bound.

\begin{figure}[htbp]
\includegraphics[width=8cm]{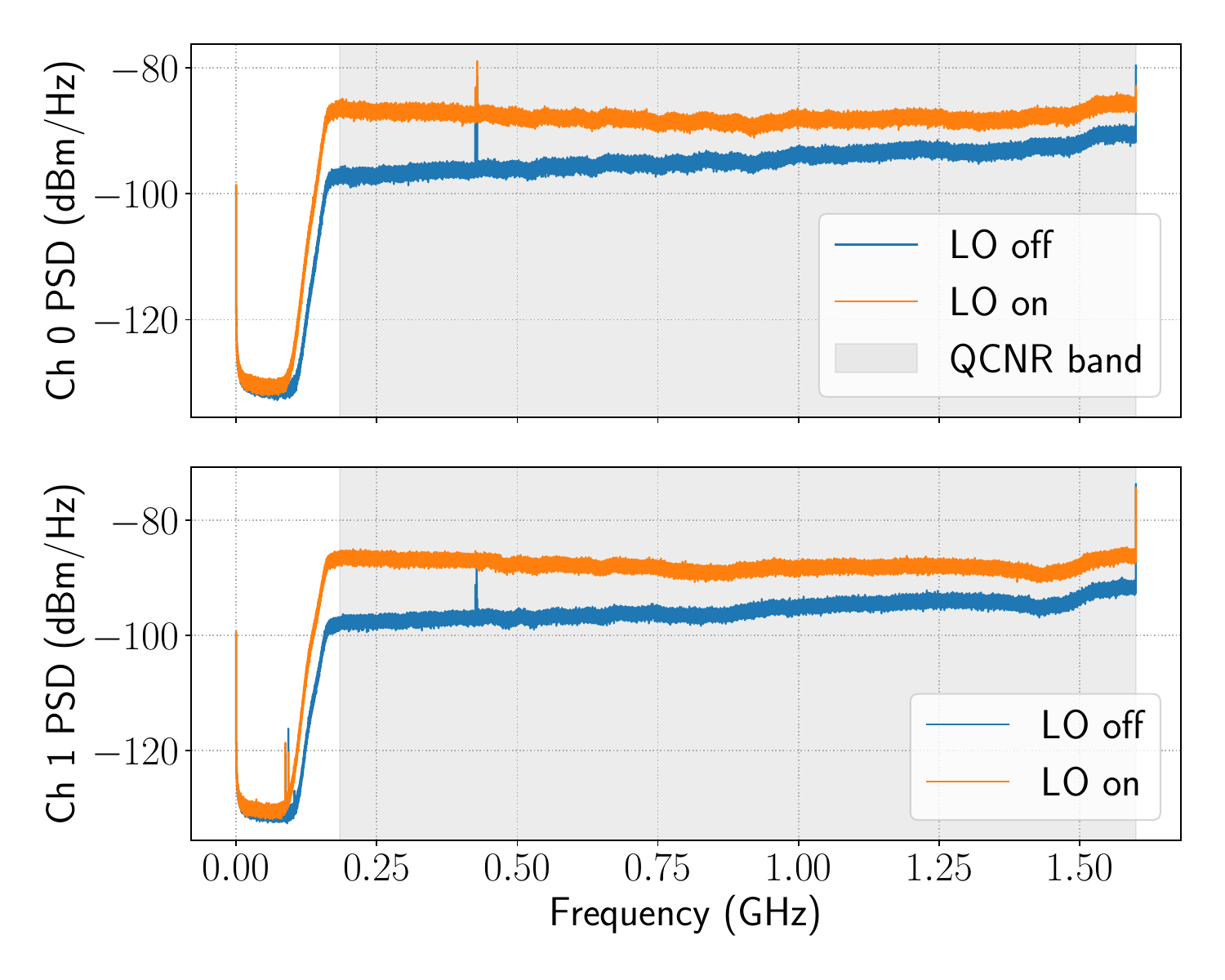}
\caption{PSD of both ADC channels that are connected to the balanced diodes. The blue curve corresponds to no laser power and the orange curve to a local-oscillator power of $8.4\,\mathrm{mW}$. The gray colored area shows over which range the PSD is integrated for QCNR.}
\label{fig:fft}
\end{figure}

PSD estimates were obtained using a one-sided periodogram with a Hann window (equivalent noise bandwidth, $\mathrm{ENBW} = 1.5\,\Delta f$). The fast Fourier transform (FFT) length was set to $N = 2^{20}$, with a sampling rate of $f_s = 3.2\,\mathrm{GS/s}$, resulting in a frequency bin width of $\Delta f = f_s / N$. Each spectrum was averaged over $M = 64$ blocks, and the error is $\pm 1.96/\sqrt{M}$ confidence interval on the dB scale, which is not shown due to being negligibly small.

To further characterize the QRNG, we measure the Husimi $Q$ function and the variance dependence on LO power. With the vacuum-unit axes from Eqs.~\eqref{eq:varfit}–\eqref{eq:deltas-infl}, the 2D histogram $H_{mn}$ over bin area $\delta'_X\delta'_P$ is converted via
\begin{equation}
\widehat{Q}_{mn} := \frac{2\,H_{mn}}{N\,\delta'_X\delta'_P},\qquad 
\sum_{m,n}\widehat{Q}_{mn}\,\frac{\delta'_X\delta'_P}{2}=1.
\end{equation}

\begin{figure}[htbp]
\includegraphics[width=8cm]{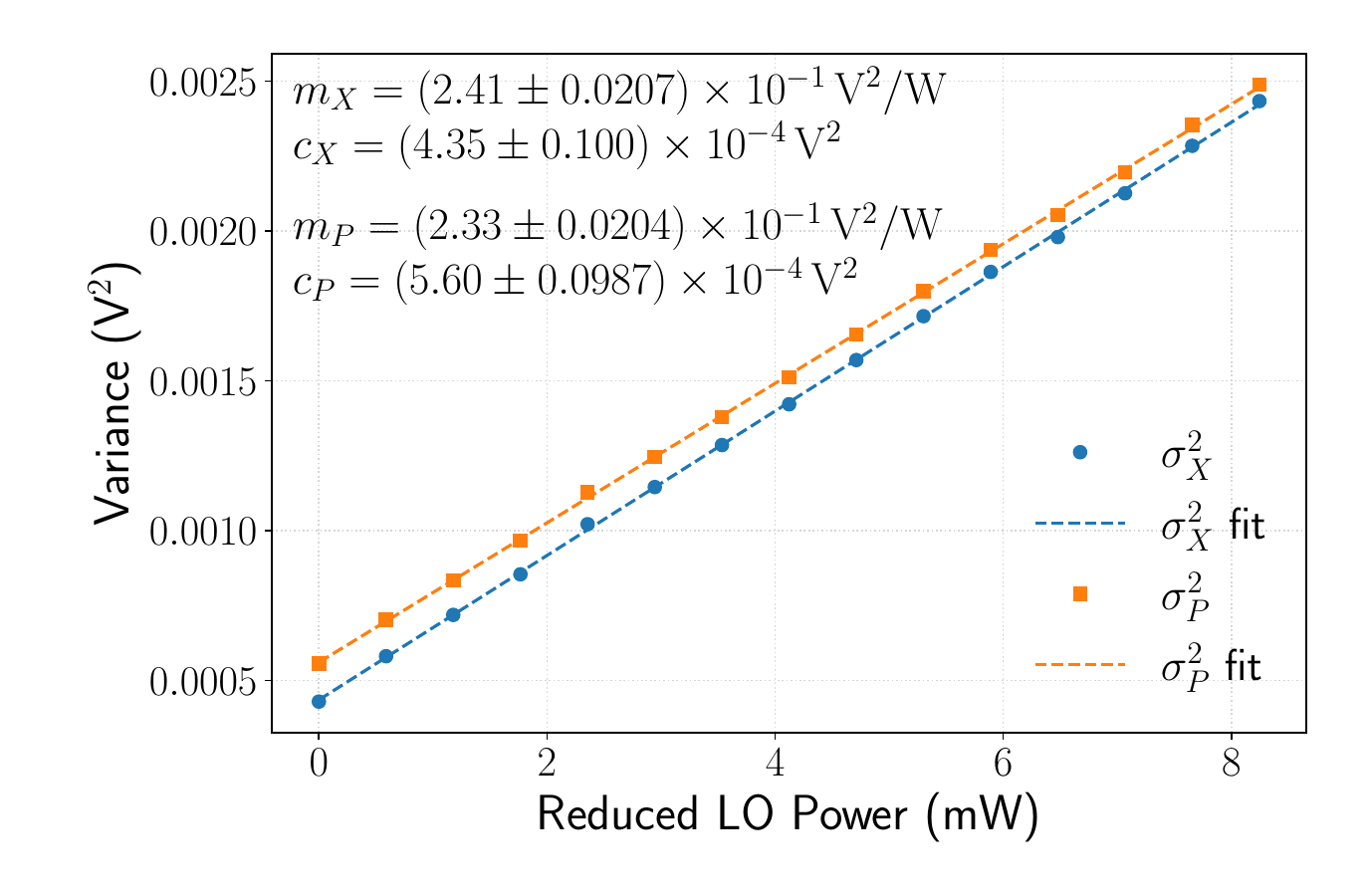}
\caption{The dependence of both quadrature measured signal variance on the local oscillator laser power.}
\label{fig:lo_cal}
\end{figure}

The calibration curve is shown in Fig.~\ref{fig:lo_cal}. The resulting fit values are:
\begin{align*}
m_X=(2.41\pm0.0207){\times}10^{-1}\,\mathrm{V}^2/\mathrm{W}\:\mathrm{(95\%\:CI)},\\
c_X=(4.35\pm0.100){\times}10^{-4}\,\mathrm{V}^2\:\mathrm{(95\%\:CI)},\\
m_P=(2.33\pm0.0204){\times}10^{-1}\,\mathrm{V}^2/\mathrm{W}\:\mathrm{(95\%\:CI)},\\
c_P=(5.60\pm0.0987){\times}10^{-4}\,\mathrm{V}^2\:\mathrm{(95\%\:CI)}
\end{align*}
At the operating LO power $P_{\rm LO}^{(\mathrm{det})}={\,8.4~\mathrm{mW}\,}$ we obtain dimensionless vacuum-unit variances (e.g., $\sigma_X^2\simeq {0.616}$, $\sigma_P^2\simeq {0.647}$) and the conservative resolutions
\begin{equation}
\delta'_X={\,0.0300\,}\quad\text{and}\quad \delta'_P={\,0.0319\,}
\end{equation}
(with uncertainties), computed via Eqs.~\eqref{eq:alphas}–\eqref{eq:deltas-infl}. These resolutions are then used to convert the voltage histogram to the phase-space representation shown in Fig.~\ref{fig:quad}.
From these resolutions,
\begin{equation}
\begin{aligned}
h_{\min}^{(1)}=-\log_2\Bigl(\min\bigl\{\tfrac{\delta'_X\delta'_P}{2\pi},\,1\bigr\}\Bigr)\\
=\log_2\Big(\frac{2\pi}{0.0300\times 0.0319}\Big)\approx 12.681~\text{bits/round}.
\end{aligned}
\end{equation}
For $n_{\rm eff}=k_{\rm in}/b=2880/24=120$ rounds per block, this gives
\begin{equation}
H_{\min}(Z^{n_{\rm eff}}|E)\ge 1521~\text{bits/block}.
\end{equation}

\begin{figure}[htbp]
\includegraphics[width=8cm]{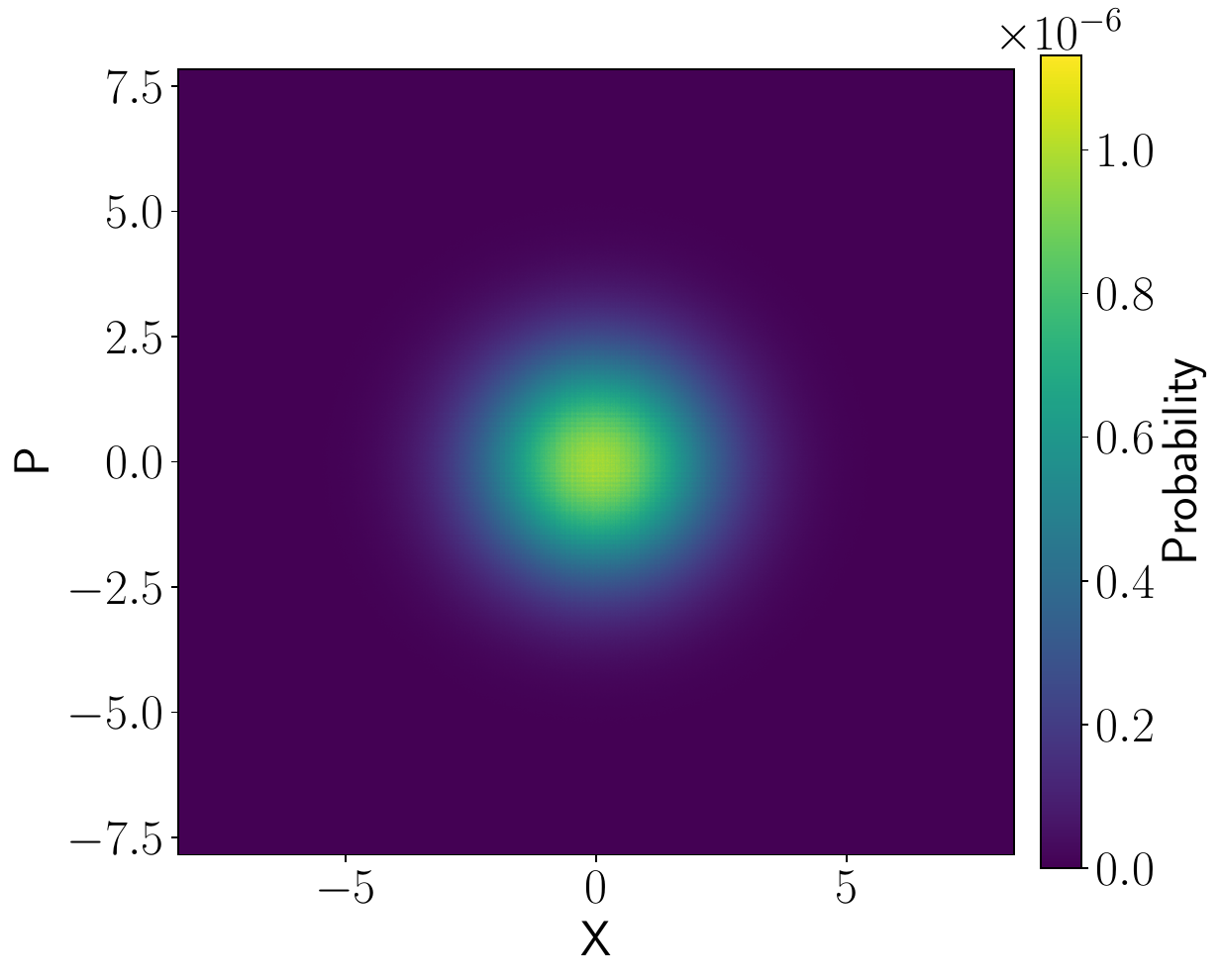}
\caption{Husimi–Q distribution in quadrature space: horizontal axis shows the in-phase quadrature X and vertical axis the orthogonal quadrature P. The color scale encodes the normalized probability per pixel.}
\label{fig:quad}
\end{figure}
We verified that the empirical bin frequencies $f_{mn}$ satisfy
$f_{mn}\le \min\bigl\{\delta'_X\delta'_P/(2\pi),\,1\bigr\} + \Delta$ with high confidence,
using a multiplicative Chernoff bound at significance $\alpha=10^{-6}$
together with a union bound over all bins $(m,n)$; violations signal a calibration error in $(\alpha_{X/P},\delta'_{X/P})$.

Several discrepancies can be seen compared to Ref.~\cite{Avesani_2018}. Firstly, we show an increased variance value in vacuum units. The most likely reason is the difference in set-ups. Ref.~\cite{Avesani_2018} immediately measures the quadrature signals from the diode with an oscilloscope that has a low enough measurement range to capture the low diode voltages. In our case, the lowest voltage that can be measured with the ADC without measurement quality deterioration is 0.5 V peak-to-peak (from -0.25 to 0.25 V), therefore, we have used an amplifier, which may have different noise characteristics, potentially causing an increase in the noise.

The second difference is that in our case the power dependencies of both quadratures do not completely overlap. This is also related to the difference in set-ups. Together with amplifiers we used hardware filters to reduce the complexity of the FPGA implementation. While the components models were the same between the two quadratures, there can still be variations between the same devices in insertion loss, gain, and other characteristics that can cause a reduction in the signal. To increase the overlap, we used the EVOAs before the photodiodes to reduce the power of one of the quadratures.

With the power calibration complete, we obtain $H_{\min}\ge 1521$ per block; we extract $\ell=1272$ bits per block, i.e.\ $(\ell/k_{\rm in})\approx 0.442$, using a $1272\times 2880$ Toeplitz matrix. The next step is to confirm whether the Toeplitz extraction algorithm is flattening the distribution seen in Fig.~\ref{fig:quad}. This can be partially done by looking at the autocorrelation of the numbers before and after extraction, which is shown in Fig.~\ref{fig:corr}.

The net extracted bit rate is
\begin{equation}
R_{\rm net}\;=\;\frac{\ell}{n_{\rm eff}}\,f_s^{(\mathrm{eff})},
\label{eq:R_net}
\end{equation}
where $n_{\rm eff}$ is the number of retained rounds after enforcing factorization (by fixed-ratio decimation or an invertible whitener). For decimation by $D$, $n_{\rm eff}=n/D$ and $f_s^{(\mathrm{eff})}=f_s/D$; for an orthonormal (unit-Jacobian) transform $n_{\rm eff}=n$ and $f_s^{(\mathrm{eff})}=f_s$. Clipped samples are dropped, so $n_{\rm eff}$ already counts valid rounds and no separate $\Delta_{\rm clip}$ term is needed. The measured clipping fraction at the operating point is $f_{\rm clip}=2.3\times10^{-6}$ ($95\%$ CI).

\begin{figure}[htbp]
\includegraphics[width=8cm]{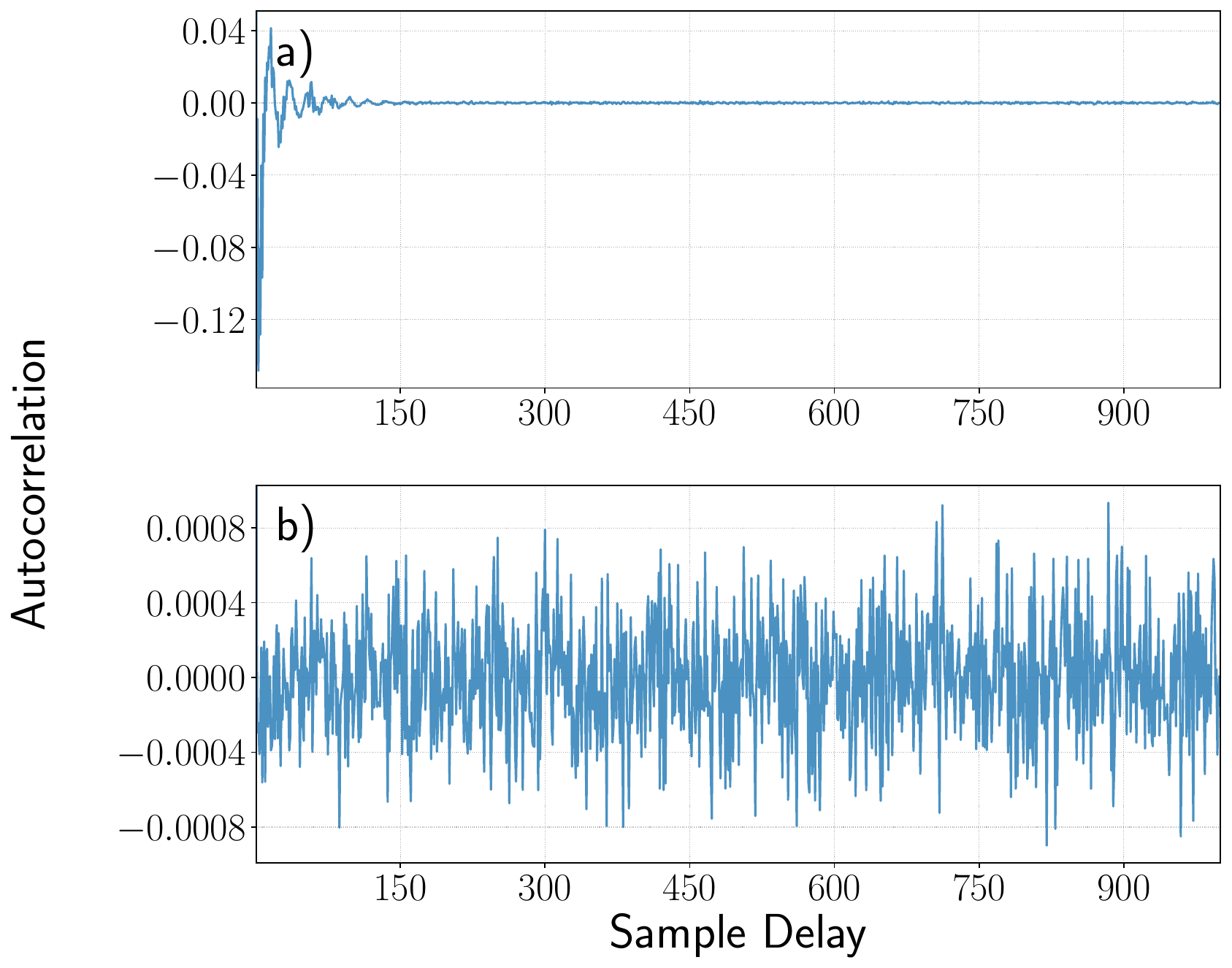}
\caption{Autocorrelation comparison between a) raw Gaussian distribution data and b) random data extracted by Toeplitz extraction algorithm. For both plots the autocorrelation and lag is applied in terms of 1000 12-bit samples. The optical power in both cases is {8.4 mW}. A total of $10^7$ samples were taken for calculation.}
\label{fig:corr}
\end{figure}

There is a clear correlation observed in the random data before extraction, which is also seen in other QRNGs~\cite{Avesani_2018, Zheng_2018, Li_2024}. It is caused partially by noise introduced by the sampling device~\cite{Avesani_2018}, however, the most significant contributor in this case is the sample rate related to the filtering that we apply. Band-limitation induces temporal correlations. We estimate the normalized ACF $\rho(k)$ and define the integrated autocorrelation time
\begin{equation}
\begin{aligned}
\tau_{\mathrm{int}} := \max\Bigl\{1,\; 1+2\sum_{k=1}^{K^\star}\rho(k)\Bigr\},\\
K^\star = \min\{k\ge 1:\,|\rho(k)|< 2/\sqrt{N_{\rm samp}}\}.
\end{aligned}
\end{equation}
In our data we obtain
\begin{equation*}
\tau_{\mathrm{int}}={\,2.68 \pm 0.00011}\:\mathrm{(95\%\:CI)},
\end{equation*}
using $B=10^4$ bootstrap resamples with block length $\ell_B=10^5$ to keep the calculation time reasonable. After extraction, residual autocorrelations are consistent with statistical fluctuations (we do not reduce the security bound by $\tau_{\mathrm{int}}$ since it can be minimized by decimating the sampling frequency).
We additionally inspected higher-order autocorrelations (Ljung–Box) up to a fixed lag window, finding no statistically significant residual structure at the 95\% level.

\begin{figure}[htbp]
\includegraphics[width=8cm]{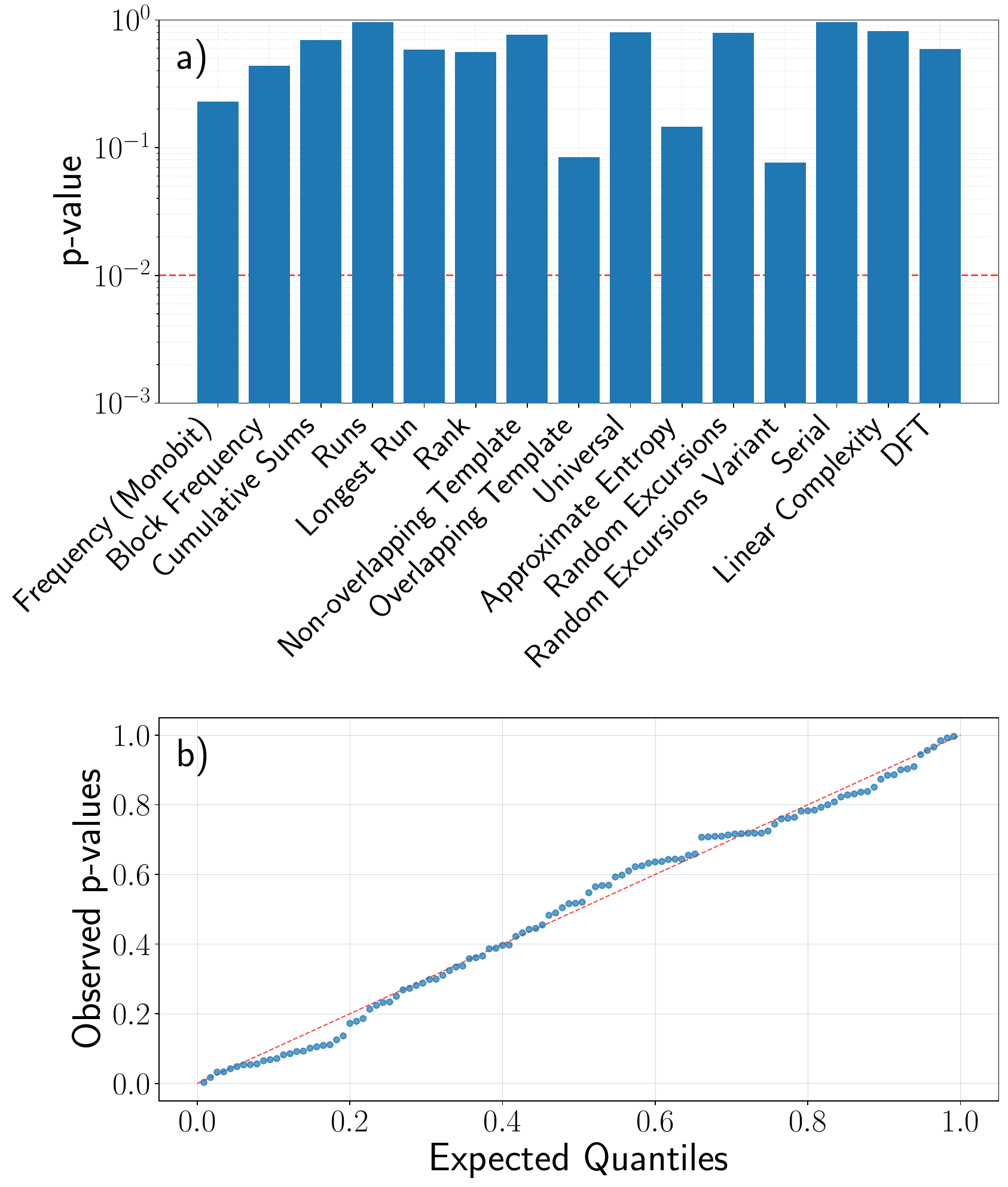}
\caption{Statistical test results for quantum random number generator output. (a) NIST Statistical Test Suite results showing p-values for individual randomness tests on a logarithmic scale. The dashed red line indicates the significance threshold ($\alpha = 0.01$). Tests with single occurrences show individual p-values, while tests with multiple instances show Kolmogorov-Smirnov test p-values comparing the distribution of individual p-values against a uniform distribution. (b) Quantile-quantile plot of Dieharder test suite p-values against expected uniform quantiles. Points following the diagonal reference line (dashed red) indicate p-values consistent with a uniform distribution expected from a true random source.}
\label{fig:tests}
\end{figure}

Lastly, statistical quality of the generated random numbers was evaluated using both the GNU Dieharder suite~\cite{dieharder} and the NIST SP 800-22 Statistical Test Suite~\cite{nist80022}. For Dieharder, all available tests were executed on binary input files of size 16 GB, with the generator set to file input and the analysis configured to apply the Kolmogorov–Smirnov test with two rejections for failure determination and full reporting of p-values. We also sorted the resulting p-values from the Dieharder tests and plotted them on a Q-Q plot. For the NIST tests, we run it configured to do 1000 iterations with 100 p-values per iteration, using sequences of 1,000,000 bits each. The results from the tests are shown in Fig.~\ref{fig:tests}. As can be seen, the NIST tests are all passing and the dieharder test p-values are all uniformly distributed along the diagonal reference line, indicating that no patterns were detected. It should be noted that these tests are only statistical, meaning that they do not show whether there is true randomness. These tests only can tell whether a source is not random.

\flushbottom

\section{Conclusions}
We have presented a fully hardware integrated, source device independent quantum random number generator that uses heterodyne detection of vacuum fluctuations. We realize the random number acquisition and Toeplitz hashing based randomness extraction entirely on FPGA logic. We also do an extensive security analysis, which shows that the source-device independent model that we consider does not rely on a model of the source, but rather on the properties of the measurement itself, which is useful in the case where the source characteristics drift or when they may be adversarially controlled.

We demonstrate a system with a real-time random number generation rate of 33.92 Gbit/s with all the security bounds taken into account. We find that with analog filtering the classical noise does not significantly impact the extractable amount of data, rather the ADC ENOB has the largest impact. Therefore, by using an ADC with higher ENOB, significantly higher QRNG rates could be achieved. We additionally explore additional security limiting factors, such as ADC clipping, inter-quadrature correlation and time autocorrelation. We find that these quantities do not significantly impact the security and therefore we only present them as diagnostics.

Most importantly, this set-up shows that high bandwidth true RNG can be achieved with commonly available components. Moreover, the PCIe interface allows for high speed transfer of the random numbers into another platform, allowing easy access to the QRNG for different applications.

\section*{Acknowledgments}
The authors acknowledge the Bundesministerium für Bildung und Forschung in the frame of the project QR.N (contract no. 16KIS2201) and EPSRC Grant No. EP/S030751/1. A. G. acknowledges funding by the state of North Rhine-Westphalia through the EIN Quantum NRW program. Data is made available upon request.

\bibliography{references} 

\end{document}